\documentclass[11pt,reqno]{amsart}
\usepackage{amsmath, amsthm, amssymb, bbm}
 \usepackage[foot]{amsaddr}
\usepackage{enumitem}

\usepackage{fullpage}
\usepackage[many]{tcolorbox}
\usepackage{xcolor}
\usepackage{wasysym}
\usepackage{mathtools}
\usepackage{scalerel}
\usepackage{bbold}
\usepackage{braket}
\usepackage{graphicx}

\usepackage{tikz-cd}
\tikzset{
  symbol/.style={
    draw=none,
    every to/.append style={
      edge node={node [sloped, allow upside down, auto=false]{$#1$}}}
  }
}
\usepackage{tikz}
\usetikzlibrary{cd}
\usetikzlibrary{calc}
\usetikzlibrary{decorations,decorations.pathreplacing,decorations.markings,decorations.pathmorphing}

\tikzset{super thick/.style={line width=3pt}}
\tikzstyle{far>}=[decoration={markings, mark=at position 0.75 with {\arrow{>}}}, postaction={decorate}]
\tikzstyle{mid>}=[decoration={markings, mark=at position 0.55 with {\arrow{>}}}, postaction={decorate}]
\tikzstyle{mid<}=[decoration={markings, mark=at position 0.55 with {\arrow{<}}}, postaction={decorate}]
\tikzset{super thick/.style={line width=3pt}}
\tikzstyle{far>}=[decoration={markings, mark=at position 0.75 with {\arrow{>}}}, postaction={decorate}]
\tikzstyle{mid>}=[decoration={markings, mark=at position 0.55 with {\arrow{>}}}, postaction={decorate}]
\tikzstyle{mid<}=[decoration={markings, mark=at position 0.55 with {\arrow{<}}}, postaction={decorate}]
\tikzstyle{knot}=[preaction={super thick, white, draw}]
\tikzstyle{coupon}=[draw, very thick, rectangle, rounded corners=5pt]
\tikzset{Rightarrow/.style={double equal sign distance,>={Implies},->},
triplecd/.style={-,preaction={draw,Rightarrow}},
quadruplecd/.style={preaction={draw,Rightarrow,
shorten >=0pt
},
shorten >=1pt,
-,double,double
distance=0.2pt}}
\tikzset{
    tripleline/.style args={[#1] in [#2] in [#3]}{
        #1,preaction={preaction={draw,#3},draw,#2}
    }
}
\tikzstyle{triple}=[tripleline={[line width=.15mm,black] in
      [line width=.7mm,white] in
      [line width=1mm,black]}] 
\tikzset{
    quadrupleline/.style args={[#1] in [#2] in [#3] in [#4]}{
        #1,preaction={preaction={preaction={draw,#4},draw,#3}, draw,#2}
    }
}
\tikzstyle{quadruple}=[quadrupleline={[line width=.3mm,white] in
      [line width=.6mm,black] in
      [line width=1.2mm,white] in
      [line width=1.5mm,black]}]

\definecolor{violet}{RGB}{148,0,211}
\definecolor{DarkGreen}{RGB}{34,139,34}
\definecolor{rufous}{HTML}{A81C07}

\usepackage[plainpages=false,hypertexnames=false,pdfpagelabels]{hyperref}
\definecolor{medium-blue}{rgb}{0,0,.8}
\hypersetup{colorlinks, linkcolor={purple}, citecolor={medium-blue}, urlcolor={medium-blue}}
\newcommand{\arxiv}[1]{\href{http://arxiv.org/abs/#1}{\tt arXiv:\nolinkurl{#1}}}
\newcommand{\arXiv}[1]{\href{http://arxiv.org/abs/#1}{\tt arXiv:\nolinkurl{#1}}}

\DeclareMathOperator{\tr}{tr}

\def\semicolon{;}
\def\applytolist#1{
    \expandafter\def\csname multi#1\endcsname##1{
        \def\multiack{##1}\ifx\multiack\semicolon
            \def\next{\relax}
        \else
            \csname #1\endcsname{##1}
            \def\next{\csname multi#1\endcsname}
        \fi
        \next}
    \csname multi#1\endcsname}

\def\calc#1{\expandafter\def\csname c#1\endcsname{{\mathcal #1}}}
\applytolist{calc}QWERTYUIOPLKJHGFDSAZXCVBNM;
\def\bbc#1{\expandafter\def\csname bb#1\endcsname{{\mathbb #1}}}
\applytolist{bbc}QWERTYUIOPLKJHGFDSAZXCVBNM;
\def\bfc#1{\expandafter\def\csname bf#1\endcsname{{\mathbf #1}}}
\applytolist{bfc}QWERTYUIOPLKJHGFDSAZXCVBNM;
\def\sfc#1{\expandafter\def\csname s#1\endcsname{{\sf #1}}}
\applytolist{sfc}QWERTYUIOPLKJHGFDSAZXCVBNM;
\def\fc#1{\expandafter\def\csname f#1\endcsname{{\mathfrak #1}}}
\applytolist{fc}QWERTYUIOPLKJHGFDSAZXCVBNM;
\def\rmc#1{\expandafter\def\csname rm#1\endcsname{{\mathrm #1}}}
\applytolist{rmc}QWERTYUIOPLKJHGFDSAZXCVBNM;

\numberwithin{equation}{section}

\theoremstyle{plain}
\newtheorem{thm}[equation]{Theorem}
\newtheorem*{thm*}{Theorem}
\newtheorem{cor}[equation]{Corollary}
\newtheorem{lem}[equation]{Lemma}
\newtheorem{prop}[equation]{Proposition}
\newtheorem{question}[equation]{Question}
\newtheorem*{claim*}{Claim}

\theoremstyle{definition}
\newtheorem{defn}[equation]{Definition}

\newtheorem*{trick*}{Trick}

\newtheorem{rem}[equation]{Remark}

\makeatletter
\def\l@subsection#1#2{}
\def\l@subsubsection#1#2{}
\makeatother

\title{Quantum cellular automata and categorical dualities of spin chains}
\author{Corey Jones$^{1}$}
\address{$^{1}$ Department of Mathematics, North Carolina State University, Raleigh, NC 27695, USA}
\author{Kylan Schatz$^{1}$}
\author{Dominic J.~Williamson$^{2,*}$}
\thanks{${}^*$ Current address: IBM Quantum, IBM Almaden Research Center, San Jose, CA 95120, USA}
\address{$^{2}$ School of Physics, University of Sydney, Sydney, New South Wales 2006, Australia}

\begin{document}

\maketitle

\begin{abstract}
Dualities play a central role in the study of quantum spin chains, providing insight into the structure of quantum phase diagrams and phase transitions. In this work, we study categorical dualities, which are defined as bounded-spread isomorphisms between algebras of symmetry-respecting local operators on a spin chain. We consider generalized global symmetries that correspond to unitary fusion categories, which are represented by matrix-product operator algebras. A fundamental question about dualities is whether they can be extended to quantum cellular automata on the larger algebra generated by all local operators in the the unit matrix-product operator sector. For on-site representations of Hopf algebra symmetries, this larger algebra is the usual tensor product quasi-local algebra. 
We present a solution to the extension problem using the machinery of Doplicher-Haag-Roberts bimodules. 
Our solution provides a crisp categorical criterion for when an extension of a duality exists. We show that the set of possible extensions form a torsor over the invertible objects in the relevant symmetry category. As a corollary, we obtain a classification result concerning dualities in the group case.
\end{abstract}

\tableofcontents

\section{Introduction}

Quantum cellular automata (QCA) are a class of quantum operations that capture the fundamental properties of unitarity and locality~\cite{QC,QC2,QC3}. 
The problem of classifying QCA up to finite-depth unitary circuits has been the focus of much attention~\cite{MR2890305,https://doi.org/10.48550/arxiv.quant-ph/0405174,MR4381173,MR4103966,MR4309221,https://doi.org/10.48550/arxiv.2205.09141}, due to both the fundamental nature of QCA and their wide ranging applications. These include invertible quantum phases of matter~\cite{HFH2020,JonesMetlitski21}, floquet quantum dynamics~\cite{Nahum,PoChiral,PoRadical,PotterVishwanathFidkowski18,PotterMorimoto17,Zhang2021classification,aasen2023measurement}, many-body localization~\cite{Long2024Topo}, unitary tensor network operators~\cite{IgnacioCirac2017,Sahinoglu2018,Piroli2020,Piroli21Fermionic,Ranard20}, and simulations of quantum field theories~\cite{ARRIGHI2011372,QFT1,QFT2}.

Thus far, most work has focused on QCAs embodied by bounded-spread automorphisms on the standard quasi-local operator algebra of a spin system, $A=\otimes_{\mathbbm{Z}^{n}} M_{d}(\mathbbm{C})$, formed by taking the tensor product of finite qudit degrees of freedom on the sites of a lattice.
This has led to a full classification of QCAs in one spatial dimension up to finite-depth circuits ~\cite{MR2890305} and partial classifications in higher spatial dimensions~\cite{MR4103966,MR4309221,https://doi.org/10.48550/arxiv.2205.09141}. 
There is a closely related problem of understanding \textit{dualities}, which in this context are mathematically described by bounded-spread isomorphisms between \textit{subalgebras} of the usual quasi-local algebra that interchange quantum phases of matter. Questions in this direction have attracted much recent interest \cite{Aasen_2016,Thorngren2019,Choi2021,Kaidi2022,PRXQuantum.4.020357,2304.00068,JL24,Bhardwaj2024,ma2024quantumcellularautomatasymmetric}.

Dualities can arise naturally in the presence of a global symmetry $\mathcal{C}\curvearrowright A$, which may correspond to an action of a group, a (weak) Hopf algebra, or a fusion category. 
If a pair of Hamiltonians has symmetric local terms, then these lie in the algebra $A^{\mathcal{C}}\subseteq A$ of symmetric operators. 
In this work, we use the term \textit{duality} to formally denote a bounded-spread isomorphism between subalgebras of symmetric operators $A^{\mathcal{C}}$ under some global symmetry which exchange Hamiltonians \textit{but may not extend} to bounded-spread automorphisms of the whole quasi-local algebra $A$ (QCA). The motivating example for this class of dualities is due to Kramers and Wannier, originally introduced in the context of classical statistical mechanics ~\cite{Kramers1941}. 
The bounded-spread automorphism implementing KW duality acts on the symmetric subalgebra of operators on a chain of qubits that commute with the global $\mathbb{Z}_2$ symmetry represented by $\prod_i \sigma^{x}_i$. The action of the Kramers-Wannier (KW) duality can be specified on a set of generators of this $\mathbb{Z}_2$-symmetric operator algebra via the map $\sigma^{x}_i\mapsto \sigma^{z}_{i-1}\sigma^{z}_{i}$ and $ \sigma^{z}_{i}\sigma^{z}_{i+1}\mapsto \sigma^{x}_{i}$. 

In recent years, there has been significant interest in extending results from the setting of conventional global symmetries, which are represented by tensor products of single-site unitary operators, to categorical global symmetries, which are represented by matrix product operators (MPOs)~\cite{PhysRevLett.93.070601,MR3614057,Vanhove2018,Aasen_2016,Thorngren2019}.  In this work, we study this flavor of duality in one spatial dimension. The resulting algebra generated by symmetric local operators (preserving the unit MPO sector) is called a \textit{fusion spin chain}.
Recently, the idea of KW duality has been generalized in this direction to categorical symmetries, where matrix-product operators have been found to implement dualities based on any invertible bimodule categories between fusion categories~\cite{Aasen_2016,Williamson2017SET,https://doi.org/10.48550/arxiv.2008.08598,PRXQuantum.4.020357}. This suggests a deep connection between categorical dualities and the algebraic structure of fusion categories. 

A key feature of non-trivial dualities following the KW paradigm is that bounded-spread isomorphisms of the symmetric operators \textit{should not be} locally extendable to the whole spin chain (or in the case of non-unital MPOs, to the local edge-restricted sector). 
We call such an extension a \textit{spatial implementation.} 
Spatial implementations are themselves ordinary 1+1D QCAs in the case of group or Hopf algebra symmetries. 
More generally, spatial implementations are \textit{edge-restricted} QCA. 
These extensions lead naturally to the following question, whose answer sheds light on the relationship between QCA and categorical dualities:

\begin{question}\label{Original question}
 When do bounded-spread isomorphisms between symmetric local operator algebras (i.e. categorical dualities) extend to QCA defined on the full (or edge-restricted) local operator algebras, and how are the possible extensions characterized?
\end{question}

To answer this question, we first establish a precise reformulation in terms of \textit{abstract fusion spin chains} $A(\mathcal{E},X)$, built from a fusion category $\mathcal{E}$ and an object $X\in \mathcal{E}$,  by defining local algebras as endomorphism algebras of tensor powers of $X$ (see Section \ref{sec:fusionspinchians}). A fusion category $\mathcal{C}$ acting on a spin chain via MPOs is specified by the data of an indecomposable right $\mathcal{C}$-module category $\mathcal{M}$, and an object $X\in \mathcal{C}^{*}_{\mathcal{M}}$. Then, the symmetric operator algebra is isomorphic to the fusion spin chain $A(\mathcal{C}^{*}_{\mathcal{M}}, X)$. If we view $\mathcal{M}$ as a left $\mathcal{C}^{*}_{\mathcal{M}}$-module category, then rewriting $\mathcal{E}:=\mathcal{C}^{*}_{\mathcal{M}}$ we can reformulate the data of a categorical symmetry as an abstract fusion spin chain $A(\mathcal{C}, X)$ together with an indecomposable left $\mathcal{C}$-module category $\mathcal{M}$. The inclusion $A(\mathcal{E}, X)\hookrightarrow A(\mathcal{E}, X)_{\mathcal{M}}$ (see Definition~\ref{spatialrealization}) corresponds to the inclusion of symmetric operators into all (edge-restricted) local operators, which we call a \textit{spatial realization} of the chain $A(\mathcal{E}, X)$. Question \ref{Original question} can be given a precise mathematical formulation as follows: 

\begin{question}\label{reformulated question}
Given a bounded-spread isomorphism $\alpha: A(\mathcal{C}, X)\rightarrow A(\mathcal{D}, Y)$ between abstract fusion spin chains, when does it extend to a QCA between spatial realizations $A(\mathcal{C}, X)_{\mathcal{M}}\rightarrow A(\mathcal{D}, Y)_{\mathcal{N}}$, and how are these extensions characterized?
\end{question}

In this form, we can address the problem using the machinery of Doplicher-Haag-Roberts (DHR) bimodules, introduced in \cite{2304.00068} and generalizing the earlier work of \cite{MR1463825}. Associated to an abstract spin system $A$ is a braided C*-tensor category $\text{DHR}(A)$, and associated to any bounded-spread isomorphism $\alpha: A\rightarrow B$ is a braided equivalence $\text{DHR}(\alpha): \text{DHR}(A)\cong \text{DHR}(B)$. Intuitively, if we view an abstract spin chain as operators living at a cut boundary of a 2+1D topologically ordered spin system (as in \cite{2307.12552}) then DHR bimodules capture the structure of the bulk topological order (see also \cite{PhysRevB.107.155136, inamura202321dsymmetrytopologicalorderlocalsymmetric}). For a fusion spin chain $A(\mathcal{C}, X)$, there is a braided equivalence $\text{DHR}(A(\mathcal{C}, X))\cong \mathcal{Z}(\mathcal{C})$, where the latter denotes the Drinfeld center of $\mathcal C$. 
We now set up the statement of our main theorem. 
For any indecomposable $\mathcal{C}$-module category $\mathcal{M}$, let $Z(\mathcal{M})\in \mathcal{Z}(\mathcal{C})$ denote the canonically associated Lagrangian algebra \cite{MR2669355}. Then we have the following (which is a special case of Theorem \ref{mainthmalt}):

\begin{thm}\label{mainthm}
    Let $\alpha: A(\mathcal{C}, X)\rightarrow A(\mathcal{D}, Y)$ be a bounded-spread isomorphism between abstract fusion spin chains. Then for any indecomposable module categories $\mathcal{M}$ and $\mathcal{N}$ of $\mathcal{C}$ and $\mathcal{D}$ respectively, spatial implementations $A(\mathcal{C}, X)_{\mathcal{M}}\rightarrow A(\mathcal{D}, Y)_{\mathcal{N}}$ of $\alpha$ are in bijective correspondence with algebra isomorphisms $\text{DHR}(\alpha)(Z(\mathcal{M}))\cong Z(\mathcal{N})$ in $\mathcal{Z}(\mathcal{D})$. As a consequence:

    \medskip
    
    \begin{enumerate}
    \item 
    If $\text{DHR}(\alpha)(Z(\mathcal{M}))$ is not isomorphic to $Z(\mathcal{N})$, then $\alpha$ has no spatial implementation.

    \medskip
    
    \item 
    If $\text{DHR}(\alpha)(Z(\mathcal{M}))\cong Z(\mathcal{N})$, then the spatial implementations of $\alpha$ form a torsor over $\text{Aut}(Z(\mathcal{M}))\cong \text{Inv}(\mathcal{C}^{*}_{\mathcal{M}})$
    \end{enumerate}
\end{thm}

For the Kramers-Wannier example described above, the fusion category is $\mathcal{C}=\mathcal{D}=\text{Rep}(\mathbbm{Z}_{2})$, $X$~is the left regular representation of $\mathbbm{Z}_{2}$, and the $\text{DHR}$ category of the symmetric spin chain is the toric code modular tensor category $\mathcal{Z}(\text{Hilb}_{f.d.}(\mathbbm{Z}_2))$. The module category in question is the standard fiber functor, which corresponds to the electric Lagrangian $L=1\oplus e$. The induced action of the Kramer-Wannier automorphism $\alpha$ is the $e\leftrightarrow m$ swap symmetry, which does not preserve $L$. 
By the above theorem, $\alpha$ admits \textit{no} QCA extension to the whole quasi-local algebra $A=\otimes_{\mathbbm{Z}} M_{2}(\mathbbm{C})$.

\medskip

We also use our formalism to make contact with the literature concerning symmetric QCA~(sQCA). In this setting, we consider a finite gloabl symmetry group $G$ acting faithfully by on-site unitaries in a self-dual representation. Then the symmetric operators $A^{G}$ have DHR category equivalent to the Drinfeld center $\mathcal{Z}(\text{Rep}(G))$. 
We immediately recover an $H^{2}(G,U(1))$ index (c.f. \cite{PhysRevLett.124.100402}) as a consequence of our framework (see Theorem~\ref{thm:sQCA}).

In \cite{JL24}, a numerical index for bounded spread automorphisms on fusion spin chains was introduced which generalizes the GNVW index \cite{MR2890305}. Combining this index with the DHR action, we obtain a homomorphism $\text{Ind}\times \text{DHR}$ from the group of dualities on an abstract fusion spin chain over a fusion category $\mathcal{C}$ to $\mathbbm{R}^{\times}\times \text{Aut}_{br}(\mathcal{Z}(\mathcal{C}))$. 

For an on-site action of a finite group $G$, letting $A^{G}$ denote the quasi-local algebra of symmetric operators, we have a braided equivalence $\text{DHR}(A^{G})\cong \mathcal{Z}(\text{Rep}(G))$. Thus the image of $\text{Ind}\times \text{DHR}$ is in $\mathbbm{R}^{\times}\times \text{Aut}_{br}(\mathcal{Z}(\text{Rep}(G)))$. The group of braided autoequivalences of the center of a group-like fusion category is well-studied for many classes of G, see for example \cite{MR2677836, MR3778972}. Recently, a complete classification of abelian $G$-dualties for on-site regular actions was obtained \cite{ma2024quantumcellularautomatasymmetric}. 
Using our results above, we obtain a similar result which also applies to non-abelian groups.


\begin{cor}\label{classificationcor}
For any finite group $G$, a pair of G-dualities in the regular representation differ by a locally symmetric finite-depth circuit if and only if they have the same numerical index and induce the same braided monoidal autoequivalence (up to natural isomorphism) on $\mathcal{Z}(\text{Rep}(G))$.
\end{cor}

The above corollary completely classifies the equivalence of $G$-dualities, including which dualities are trivial. However, it falls short of providing a closed form description of the image of $\text{Ind}\times \text{DHR}$ in $\mathbbm{R}^{\times}\times \text{Aut}_{br}(\mathcal{Z}(\text{Rep}(G)))$. We leave this question to future work. 

\bigskip

Our results point to a number of directions for future work. 
First, what is the full classification of bounded-spread isomorphisms on local algebras for fusion spin chains? Using our results, we expect to be able to reduce this problem to the classification of edge-restricted QCAs in one spatial dimension, which is an open problem. 
Second, can our results be applied to characterize symmetries of topological orders in two dimensions via their action by dualities on the holographic boundary (for example, in the sense of \cite{2307.12552})?
Third, can techniques from the theory of operator algebras, subfactors, and fusion categories be applied to classify dualities and QCAs in higher spatial dimensions? 

\bigskip

The remainder of this work is presented as follows. 
In Section~\ref{sec:2} we introduce background material on categorical dualities of quantum spin chains.
In Section~\ref{sec:fusionspinchians} we describe dualities in the context of fusion spin chains.
In Section~\ref{sec:4} we describe categorical inclusions defined by module categories over algebra objects.
In Section~\ref{sec:5} we describe how DHR bimodules can be applied to solve the extension problem for dualities of quantum spin chains.
In Section~\ref{sec:6} we present examples that generalize the KW duality, and give our applications.

\section{Dualities for spin chains with categorical symmetry}
\label{sec:2}

In this section we introduce relevant background material. 
First, we briefly recall the algebraic approach to spin chains \cite{MR1441540}. Consider $\mathbbm{Z}\subseteq \mathbbm{R}$ as a metric space. A spin chain is specified by choosing an on-site Hilbert space $\mathbbm{C}^{d}$; for every interval $I \subseteq \mathbb Z$, we have the local Hilbert space $\otimes_{x\in I} \, \mathbbm{C}^{d}$ and the algebra of operators localized in $I$:
$$A_{I}:=\otimes_{x\in I} \, M_{d}(\mathbbm{C}).$$ 

For $I\subseteq J$, we have natural inclusions $A_{I}\hookrightarrow A_{J}$ defined by $a\mapsto 1^{J<I}\otimes a\otimes 1^{J>I}$. The C*-completion of the union over all intervals (or more formally, the colimit in the category of *-algebras) is called the \textit{quasi-local algebra}. We call the \textit{algebraic union} the local algebra. States of the system in the thermodynamic limit are normalized, positive linear functionals on $A$, while Hamiltonians can be expressed as (unbounded) derivations on $A$ (see \cite{MR887100,MR1441540}).

A \textit{local Hamiltonian} is an assignment of a self-adjoint element $H_{I}\in A_{I}$ to every interval $I \subseteq \mathbb Z$ for which there exists an $R\ge 0$ such that if $I$ has length greater than $R$, then $H_{I}=0$. This defines an (unbounded) derivation on the local algebra:
$$\delta_{H}(a):=\sum_{I} [H_{I},a],$$ 

\noindent which is well-defined since the above summation contains only finitely many non-zero terms. 

The primary motivation for this paper is to understand a class of \textit{categorical dualities} between spin systems. Broadly speaking, a duality between two physical theories is a mapping between the observables of the theories that intertwines the relevant physics (state, dynamics, etc.). However, interesting dualities should be equivalences that are in some salient sense ``not obvious". In the case of spin systems dualities are generally considered in the context of symmetries, whereby we have a locality-preserving equivalence on the symmetric subsystems that does not extend to an equivalence of the whole theory (for example, see~\cite{doi:10.1080/00018732.2011.619814} or for the categorical context ~\cite{PRXQuantum.4.020357,Aasen_2016})
We give a mathematical formulation of this in 1D in terms of fusion spin chains.

First we set up some basic terminology. We need a natural notion of a locality-preserving equivalence of spin systems, which is captured by the idea of bounded-spread isomorphisms (or quantum cellular automata).

\begin{defn} If $A, B$ are the local algebras of spin chains on $\mathbbm{Z}$, a \textit{bounded-spread isomorphism} is a $*$-isomorphism $\alpha:A\rightarrow B$ such that there exists an $R>0$ so that $\alpha(A_{I})\subseteq B_{I^{+R}}$ for every interval $I$. A bounded-spread isomorphism $\alpha:A\rightarrow A$ is called a \textit{quantum cellular automaton} (QCA).
\end{defn}

In the above definition $I^{+R}$ denotes the $R$-neighborhood of the interval $I$ in $\mathbbm{Z}$. A bounded-spread isomorphism is a natural mapping between the observables of the two spin chains that \textit{uniformly} preserves locality.

\begin{defn}\label{derivation}
If $H,K$ are local Hamiltonians on spin chains with local algebras $A$ and $B$ respectively, an equivalence between them is a bounded-spread isomorphism $\alpha: A\rightarrow B$ such that $\alpha(\delta_{H}(\alpha^{-1}(b)))=\delta_{K}(b)$ for all local operators $b\in B$. We denote this by $\alpha(H)=K$.
\end{defn}

By the definition of bounded-spread isomorphism, for any local Hamiltonian $H=\{H_{I}\}$ on $A$, $\alpha(H):=\{\alpha(H)_{I}:=\alpha(H_{I^{-R}})\}$ is a local Hamiltonian on $B$ and $\alpha$ implements an equivalence between these. Thus given any bounded-spread isomorphism $\alpha: A\rightarrow B$ and any local Hamiltonian $H$ on $A$, $\alpha$ always implements an equivalence between $H$ and \textit{some} local Hamiltonian on $B$. This suggests that it could be useful to shift our focus to the groupoid of local algebras of spin chains and bounded-spread isomorphisms between them.

\medskip

\textbf{Example: Finite depth quantum circuits}. An obvious class of QCA are the \textit{finite-depth quantum circuits}, (abbreviated FDQC). Recall a depth-one circuit is given by specifying a uniformly bounded partition $\mathcal{P}$ of $\mathbbm{Z}$ into intervals, and for $I\in \mathcal{P}$, choosing a unitary $u_{I}\in A_{I}$. Then define a QCA acting on the local operator $x$ by

$$\alpha(x):=\left(\prod_{I\in \mathcal{P}} u_{I} \right) x \left(\prod_{I\in \mathcal{P}}u^{*}_{I}\right).$$

\noindent Note that by locality, this infinite product is in fact finite since all but finite many of the $u_{I}$ commute with $x$.

A \textit{finite-depth quantum circuit} (FDQC) is a finite composition of depth one circuits. FDQC are QCA that can be locally implemented. It is easy to show that FDQC is a normal subgroup of QCA. FDQC equivalence for states is used to characterize long-range entanglement classes~\cite{PhysRevB.82.155138}, and thus we can view QCA/FDQC as long-range entanglement classes of uniformly locality-preserving equivalences~\cite{https://doi.org/10.48550/arxiv.2205.09141}. Since FDQC are locally implemented, they are often considered ``trivial" as equivalences of Hamiltonians, which naturally leads us to consider the group $QCA/FDQC$, which has been frequently studied in the literature.

\medskip

Before moving on to giving a mathematical definition for our Kramers-Wannier type dualities, it is useful to have in hand the notion of an \textit{abstract spin chain.} This is strongly motivated by the Haag-Kastler axioms for an algebraic quantum field theory \cite{MR165864}, and is a straightforward generalization of the properties satisfied by the quasi-local algebra of a spin system \cite{MR1441540}.

\begin{defn} An abstract spin chain is a unital C*-algebra $A$ and for every interval $I$, a unital subalgebra $A_{I}\subseteq A$ such that 

\begin{enumerate}
\item 
(Isotony) If $I\subseteq J$, $A^{\circ}_I\subseteq A^{\circ}_{J}$
\item 
(Locality) If $I\cap J=\varnothing$, then $[A^{\circ}_{I},A^{\circ}_{J}]=0$
\item 
(Density) $(\cup A_{I})^{\| \cdot\|}=A$
\item 
(Weak algebraic Haag duality) There is some $R\ge 0$ such that for every interval $I$, $\{x\in A\ : [x,y]=0\ \text{ for all} y\in A_{J}, J\subseteq I^{c}\}\subseteq A_{I^{+R}}$.
\end{enumerate}

\end{defn}

If the last condition is satisfied with $R=0$, it is called \textit{algebraic Haag duality}. 
The abstract spin chains of interest in this paper always satisfy this stronger version of the last axiom. For two abstract spin chains, the definition of \textit{bounded-spread isomorphism} used above makes perfect sense unchanged. We note that if $\alpha$ is bounded spread, then it follows from weak algebraic Haag duality that $\alpha^{-1}$ is as well.

\bigskip

\noindent \textbf{Actions of fusion categories on spin chains.} We are now ready to describe dualities, in the same sense as~\cite{PRXQuantum.4.020357}. Suppose we have two spin systems with local algebras $A$ and $B$, both with a global symmetry $\mathcal{C}\hookrightarrow A$ and $\mathcal{D}\hookrightarrow B$ (think of an on-site group symmetry, or more generally, a categorical symmetry determined by an MPO, see below). Suppose further that the local terms of the Hamiltonians $H$ and $K$ are both symmetric. Then we can restrict our attention to the symmetric operators, which consist of subalgebras $A^{\mathcal{C}}\subseteq A$ and $B^{\mathcal{D}}\subseteq B$ containing the terms of $H$ and $K$ respectively. We can then ask for bounded-spread isomorphisms \textit{between the subalgebras} $A^{\mathcal{C}}$ and $B^{\mathcal{D}}$ which intertwine $H$ and $K$. This is what we call a duality.

First, we formalize what we mean by a fusion categorical symmetry and the algebras of symmetric operators, resulting in the definition of fusion spin chains and duality. Recall a unitary fusion category is a unitary tensor category (in the sense of~\cite{MR2183279}) with finitely many isomorphism classes of simple objects. For further definitions concerning fusion categories, modules categories, and related concepts, see~\cite{MR3242743}.

\begin{defn} Let $A$ be a spin system with local Hilbert space $\mathbbm{C}^{d}$ and let $\mathcal{C}$ be a unitary fusion category. An action of $\mathcal{C}$ on $A$ consists of 

\begin{enumerate}
\item 
An indecomposable, semi-simple right module category $\mathcal{M}$ of $\mathcal{C}$.
\item 
An object $X\in \mathcal{C}^{*}_{\mathcal{M}}$ which is \textit{strongly tensor generating} for $\mathcal{C}^{*}_{m}$
\item
A unitary isomorphism $\bigoplus_{i,j\in \text{Irr}(\mathcal{M})} \mathcal{M}(i, X\triangleright j)\cong \mathbbm{C}^{d}$.
\end{enumerate}

The action is called \textit{unital} if $\mathcal{M}$ is equivalent to $\text{Hilb}_{f.d}$ as C*-categories.

\end{defn}

From this data, we can construct \textit{matrix product operators} which act by quantum channels on (a sector of) the spin system, implementing the fusion rules of $\mathcal{C}$.

In general a matrix product operator is defined from a 4-valent tensor 
\begin{align}
M^{ij}_{\alpha \beta}:=
\vcenter{\hbox{\includegraphics[page=1]{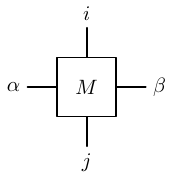}}}
\end{align}
by
\begin{align}
    \sum_{\{i\}\{j\}}  \tr( M^{i_1j_1} M^{i_2j_2}... \noindent M^{i_nj_n} ) \ket{\{i\}}\bra{\{j\}}
\end{align}
where $\{i\}=\{i_1,i_2,\dots,i_n\}$ and similarly for $\{j\}$, and 
\begin{align}
    \tr( M^{i_1j_1} M^{i_2j_2}... \noindent M^{i_nj_n} ) = 
    \vcenter{\hbox{\includegraphics[page=2]{Figures.pdf}}}
\end{align}
where we treat $M^{ij}_{\alpha \beta}$ as a matrix element of $M^{i j}$ and take a matrix product over the $\alpha, \beta$ indices of the tensors $M^{ij}_{\alpha \beta}$. From the data in Definition 2.4, one can construct for each (simple) object in $\mathcal{C}$ a matrix product operator, satisfying a family of categorical relations \cite{Sahinoglu2014,MR3614057}. 
It is in this sense that a fusion category is actually acting on the spin chain. Alternatively, one can build from an MPO a locality preserving quantum channel, which is defined on local operators via a suitable normalization of the following picture
\begin{align}
    \vcenter{\hbox{\includegraphics[page=3]{Figures.pdf}}}
\end{align}
where $M^*$ denotes the complex conjugate of $M$ with the $i$ and $j$ indices swapped. 
The categorical relations imply that these quantum channels implement a normalized version of the fusion rules of the category.

\bigskip

\noindent \textbf{Edge-restricted algebra}. For non-unital fusion categorical symmetries, it is natural to consider a non-unital ``subalgebra" of $A$, called the \textit{edge-restricted} algebra. If we pick an orthonormal basis $E_{i,j}$ for the space $\mathcal{M}(i, X\triangleleft j)$, we can define a finite graph $\mathcal{G}$ whose vertices are $\text{Irr}(\mathcal{M})$, and the edges from $a\rightarrow b$ are defined to be $E_{a,b}$. Then using the isomorphism between the onsite Hilbert space $H_{i}$ and $\bigoplus_{i,j\in \text{Irr}(\mathcal{M})} \mathcal{M}(i, X\triangleleft j)\cong \mathbbm{C}^{d}$, we have a distinguished basis of $H_{m}$ indexed by the edges of $\mathcal{G}$, for every $m\in \mathbbm{Z}$

Then for each pair of adjacent sites $(i, i+1)$, consider the local projection $P_{i}\in A_{[i,i+1]}$ defined by $P_{i}(e\otimes f)=\delta_{t(e)=s(f)}e\otimes f$, which checks the adjacency of the edges $e$ and $f$ in the graph $\mathcal{G}$. Note that $[P_{i},P_{j}]=0$ for all $i$ and $j$, and thus we can (unambiguously) consider the projection 

$$P:=\prod_{i\in \mathbbm{Z}} P_{i}.$$

\noindent We have two immediate observations:

\begin{enumerate}
\item 
This projection in general does not lie in the quasi-local C*-algebra $A$. However, if we pick a state $\phi$ on $A$ such that $\phi(P_{i}\cdot P_{i})=\phi$ for all $i$, then $P$ does define an operator on the GNS Hilbert space $L^{2}(A,\phi)$ in the von Neumann completion of $A$. If we wish to avoid picking a state, we can view $P$ as an element of the von Neumann algebra $A^{**}$, which is the von Neumann completion of $A$ in the universal Hilbert space representation of $A$ \cite{MR548728}.

\item

If we consider a finite chain on an interval $I$, then the action of the restricted version of $P$, $P_{I}$ on $H_{I}$ agrees with the action of the MPO assigned to the unit object $\mathbbm{1}\in \mathcal{C}$ away from the boundary. For this reason, it is reasonable to call the image of $P$ in a Hilbert space representation the \textit{unit MPO sector}. This is the sector on which the categorical MPO operators satisfy the required categorical relations \cite{MR3614057}.

\end{enumerate}

We can consider the sector ``corner subalgebra" of A, defined by $A^{\circ}:=PAP$ (in $A^{**}$), with local subalgebras $A^{\circ}_{I}:=PA_{I}P$. Technically speaking, the the elements $A^{\circ}_{I}$ are not properly elements of the algebra $A$ or even its C*-completion, but rather they are elements of the von Neumann completion of $A$ (either in the double dual von Neumann algebra $A^{**}$, or in any faithful Hilbert space representation of $A$ where $P$ does not act by $0$).

We can easily identify $A^{\circ}_{I}$ as follows: If $I=[i,j]$ with $j-i=n$, let $L^{i,j}_{I}$ be the Hilbert space with orthonormal basis given by the paths of length $n$ in $\mathcal{G}$ starting at $i$ and ending at $j$. Let $d^{n}_{i,j}$ denote the number of length $n$ paths from $i$ to $j$ so that $d^{i,j}_{n}=\text{dim}(L^{i,j}_{I})$. 

Then we have an isomorphism of Hilbert spaces to 

$$\left(\prod_{a\le i<b}P_{i}\right)(H_{I}
)\cong \bigoplus_{i,j\in \mathcal{G}}L^{i,j}_{I}.$$

It is easy to see from the definitions that 

$$A^{\circ}_{I}\cong \bigoplus_{i,j\in \text{Irr}(\mathcal{M})} \text{Mat}(L^{i,j}_{I})$$

\noindent where the summand $\text{Mat}_{d^{i,j}_{n}}(\mathbbm{C})$ is identified with all the linear operators on the Hilbert space $L^{i,j}_{I}$. In other words, $A^{\circ}_{I}$ is isomorphic to the subalgebra of operators on $L_{I}$ that satisfy the property that the coefficient of the matrix unit $e_{p,q}$ is $0$ if $s(p)\ne s(q)$ or $t(p)\ne t(q)$. 

Note that from this description, $A^{\circ}_{I}$ is isomorphic to a (non-unital) subalgebra of $A_{I}$, but not coherently. Nevertheless, $A^{\circ}$ has the structure of an abstract spin chain. All of the axioms are obvious expect for weak algebraic Haag duality, which in fact follows from Lemma \ref{Haagdualitylem} below.

\bigskip

\noindent\textbf{Symmetric algebra}. The operators of interest to us are the symmetric local operators. To define these, for each $Y\in \text{Irr}(\mathcal{C})$, and interval $J$, let $M^{Y}_{J}$ denote the associated matrix product operator, acting on (a sufficiently large) edge-restricted Hilbert space with periodic boundary conditions $$L^{per}_{J}:=\bigoplus_{i,j\in \text{Irr}(\mathcal{M})} L^{i,i}_{J}$$

For an interval $I<<J$, we have a unital, faithful representation $A^{\circ}_{I}\subseteq \text{Mat}(L^{per}_{J})$. Then we define

$$A^{\mathcal{C}}_{I}:=\{x\in A^{\circ}_{I}\ :\ [M^{Y}_{J},x]=0\}.$$

Notice that since the MPO are defined locally, this doesn't depend on the choice of ambient interval for $J>>I$. We also note that we need our graph $\mathcal{G}$ to be sufficiently well-connected, which is guaranteed by the indecomposability and strongly tensor generating assumptions for $X\in \mathcal{C}^{*}_{\mathcal{M}}$.

For $I\subseteq J$, we have natural inclusions $A^{\mathcal{C}}_{I}\hookrightarrow A^{\mathcal{C}}_{J}$, and taking the colimit in the category of C*-algebras, we obtain the algebra $A^{\mathcal{C}}$. Identifying $A^{\mathcal{C}}_{I}$ with its image in the colimit gives an abstract spin system. We call the algebra $A^{\mathcal{C}}$ the \textit{symmetric algebra}.

\bigskip

\bigskip

We can now formally define a version of duality for Hamiltonians generalizing Kramers-Wannier duality to the categorical setting, following~\cite{PRXQuantum.4.020357}. We can then state the problem that motivates this work.

\begin{defn}
    Let $H=\{H_{I}\}$ be a local Hamilton on the spin system $A$, and let $\mathcal{C}$ be a fusion category acting on $A$. Then we say $H$ is $\mathcal{C}$-symmetric if for every interval $I$ with $H_{I}\ne 0$, $H_{I}P=PH_{I}\ne 0$ and $H_{I}P\in A^{\mathcal{C}_{I}}$.
\end{defn}

In this case the collection $H^{\circ}:=\{H_{I}P\}$ actually defines a local Hamiltonian internally to the symmetric chain $A^{\mathcal{C}}$, which we denote $H^{\mathcal{C}}:=\{H^{\mathcal{C}}_{I}\}$

\bigskip

\begin{defn}
Let $A$ and $B$ be spin chains with local Hamiltonians $H$ and $K$ respectively. A \textit{categorical duality} between $(A,K)$ and $(B,K)$ consists of:
\begin{enumerate}
\item 
Categorical symmetries $\mathcal{C}\curvearrowright A$, $\mathcal{D}\curvearrowright B$ such that $H$ is $\mathcal{C}$-symmetric and $K$ is $\mathcal{D}$-symmetric
\item 
A bounded-spread isomorphism $\alpha: A^{\mathcal{C}}\rightarrow B^{\mathcal{D}}$ such that $\alpha(H^{\mathcal{C}})=K^{\mathcal{D}}$ (in the sense of Definition \ref{derivation}, but restricted to $A^{\mathcal{C}}$).
\end{enumerate}
\end{defn}

\bigskip

\noindent \textbf{Motivating Example: Kramers-Wannier}. Here we spell out the motivating example of Kramers-Wannier duality \cite{Kramers1941}. In this case, the acting fusion category is $\mathcal{C}=\text{Hilb}_{f.d.}(\mathbbm{Z}_{2})$. The module category is $\mathcal{M}=\text{Hilb}_{f.d}$ equipped with the trivial $\text{Hilb}_{f.d.}(\mathbbm{Z}_{2})$ action, and the dual object $X=\mathbbm{C}[\mathbbm{Z}_{2}]\in \text{Rep}(\mathbbm{Z}_{2})$ is the left regular representation.

This translates into the standard on-site spin flip action of $\mathbbm{Z}_{2}$ a chain of qubits $\otimes_{\mathbbm{Z}} \mathbbm{C}^{2}$, so that $A=\otimes_{\mathbbm{Z}} M_{2}(\mathbbm{C})$. The non-trivial symmetry operator (or MPO) can be represented by the product of on-site operators (i.e. with trivial valence bonds) as

$$U_{g}:=\prod \sigma^{x}_{i}.$$

In this case, since $\mathcal{M}$ is rank one, the action is unital, so that the edge-restricted algebra and the quasi-local algebra of the entire spin system coincide, $A^{\circ}=A$. 

The symmetric algebra is the C*-algebra generated by the operators $\sigma^{x}_{i}$ and $\sigma^{z}_{i}\sigma^{z}_{i+1}$ as $i$ ranges over $\mathbbm{Z}$, or more formally

$$A^{\mathbbm{Z}_{2}}=\text{C}^{*}-\text{alg}\{\sigma^{x}_{i}, \sigma^{z}_{i}\sigma^{z}_{i+1}\ :\ i\in \mathbbm{Z}\}.$$

Then consider the Ising Hamiltonian

$$H_{J,h}=\sum_{i\in \mathbbm{Z}} J\sigma^{z}_{i}\sigma^{z}_{i+1}+ h\sigma^{x}_{i},$$

which is $\mathbbm{Z}_{2}$-symmetric according to our definition, since all local terms lie in the subalgebra $A^{\mathbbm{Z}_{2}}$.

Now, the Kramers-Wannier duality is defined by

$$\alpha(\sigma^{x}_{i})=\sigma^{z}_{i-1}\sigma^{z}_{i}\ \ \text{and}\ \  \alpha(\sigma^{z}_{i}\sigma^{z}_{i+1})=\sigma^{x}_{i},$$

\medskip

\noindent which extends to an isomorphism $\alpha: A^{\mathbbm{Z}_{2}}\rightarrow A^{\mathbbm{Z}_{2}}.$ We see immediately from the definition

$$\alpha(H_{J,h})=H_{h,J}.$$

A key feature of KW-duality is that the equivalence between theories \textit{can only be witnessed} at the level of symmetric operators, and does not extend to a (uniformly locality-preserving) equivalence of the whole spin system~\cite{doi:10.1080/00018732.2011.619814}. Alternatively, some authors say that any unitary extension does not preserve locality, though this is not rigorously defined in the algebraic formalism for the thermodynamic limit we are using here.

This leads us to the following definition.

\begin{defn}
A bounded-spread isomorphism $\alpha: A^{\mathcal{C}}\rightarrow B^{\mathcal{D}}$ is called \textit{spatially implemented} if it admits an extension to a bounded-spread isomorphism $\widetilde{\alpha}: A^{\circ}\rightarrow B^{\circ}$.

\end{defn}

The idea is that spatially implemented bounded-spread isomorphisms are not proper dualities in the sense of KW, and are actually more of a standard equivalence. This is not literally true on the nose in full generality, since the $\widetilde{\alpha}:A^{\circ}\rightarrow B^{\circ}$ generally does not extend to $A\rightarrow B$, and thus more properly we would say 
that $\widetilde{\alpha}:A^{\circ}\rightarrow B^{\circ}$ such that $\widetilde{\alpha}(H^{\mathcal{C}})=K^{\mathcal{D}}$ is an \textit{emergent equivalence}~\cite{doi:10.1080/00018732.2011.619814}. Some may view an emergent equivalence as a kind of duality, but for us we regard it as closer to an equivalence than a KW-type duality. Indeed, if $A^{\circ}=A$ and $B^{\circ}=B$ (or in other words, the categorical action is unital, or equivalently, the module categories correspond to fiber functors), then the emergent equivalence reduces to an ordinary equivalence.

In any case, we can now give a precise mathematical statement of Question \ref{Original question}:

\bigskip

\noindent \textbf{Reformulation of Question \ref{Original question}}.\label{dualityquest} Given a bounded-spread isomorphism $\alpha: A^{\mathcal{C}}\rightarrow B^{\mathcal{D}}$, can we characterize when it is spatially implemented?

\bigskip

In the sections that follow, we present a precise and abstract statement of this problem, and a solution using the DHR bimodule category of fusion spin chains.

\section{Abstract fusion spin chains}\label{sec:fusionspinchians} In the previous section, we set up the problem of studying bounded-spread isomorphisms between algebras of operators invariant under a categorical symmetry (i.e. $A^{\mathcal{C}}$), and asked when this can be extended to the edge-restricted algebra (i.e. $A^{\circ}$). In this section, we present an abstract formulation of this problem in terms of fusion spin chains. We refer the reader to \cite{MR3242743} for definitions of fusion categories and related objects. For convenience, we assume all our fusion categories are strict.

\bigskip

\textbf{Fusion spin chains}. Given an indecomposible multi-fusion category $\mathcal{E}$ and an object $X\in \mathcal{E}$ (not necessarily simple) we construct an abstract spin chain $A(\mathcal{E},X)$ as follows:

\bigskip

\begin{itemize}
\item 
For each interval $I$, define $ A(\mathcal{C},X)_{I}:=\mathcal{E}(X^{\otimes |I|}, X^{\otimes |I|})$

\bigskip

\item 
If $I\subseteq J$, then we define the inclusion $A(\mathcal{E}, X)_{I}\hookrightarrow A(\mathcal{E}, X)_{J}$ by 

$$f\mapsto (1_{X})^{\otimes\ J<I}\otimes f \otimes (1_{X})^{\otimes\ J>I}$$

\item 
We set $A(\mathcal{E},X)=\text{colim}_{I}A(\mathcal{E}, X)_{I} $, where the colimit is taken in the category of $*$-algebras. 
\end{itemize}

\bigskip 

\noindent The commutativity property follows from functoriality of the monoidal product. Abstract spin chains constructed as above are called \textit{fusion spin chains} \cite{2304.00068,JL24}.

\bigskip

It turns out the three examples of abstract spin chains we saw in the previous section can be realized as fusion spin chains:

\bigskip

\begin{enumerate} 
\item \textbf{Concrete spin systems}. Choose the fusion category $\mathcal{E}=\text{Hilb}_{f.d.}$ of finite-dimensional Hilbert spaces, and pick the object $X:=\mathbbm{C}^{d}$. Then $A(\text{Hilb}_{f.d.}, X)$ can be identified with the ordinary spin chain $A$ constructed from $\mathbbm{C}^{d}$.

\bigskip

\item \textbf{Edge-restricted algebras}. Consider a fusion category action on an ordinary spin system $\mathcal{C}\curvearrowright A$, and let $X\in \mathcal{C}^{*}_{\mathcal{M}}$. This defines an object in the multi-fusion category
$\widetilde{X}\in \mathcal{E}=\text{Mat}_{m}(\text{Hilb}_{f.d.})$, where $m$ is the rank of $\mathcal{M}$. This is the multi-fusion category consisting of $m\times m$ matrices of finite-dimensional Hilbert spaces, with $\otimes$ given by ``matrix-multiplication" \cite{ferrer2020classifying}.
The object $\widetilde{X}$ is defined by 
$$\widetilde{X}_{i,j}:=\mathcal{M}(i, X\triangleleft j)$$

\noindent with Hilbert space structure on the components given by the composition inner product. It is easy to see that the edge-restricted abstract spin chain $A^{\circ}$ associated with our action $\mathcal{E}\curvearrowright A$ can be identified with the fusion spin chain $A(\text{Mat}_{m}(\text{Hilb}_{f.d.}), \widetilde{X})$.

\bigskip

\item 
\textbf{Symmetric algebras} Again, suppose we have a fusion category $\mathcal{C}\curvearrowright A$. Then we can identify $A^{\mathcal{C}}$ with the fusion spin chain $A(\mathcal{C}^{*}_{\mathcal{M}}, X)$ \cite{KawAnn, MR4272039}. 
\end{enumerate}

\bigskip

This leads us to define the following groupoid:

\begin{defn}
\text{Dua} is the groupoid whose
\begin{enumerate}
    \item Objects are fusion spin chains $A(\mathcal{E}, X)$.
    \item 
    Morphisms are bounded-spread isomorphisms with composition.
\end{enumerate}
\end{defn}

Our discussion above has motivated the definition of $\text{Dua}$, whose objects are \textit{abstract} spin systems of symmetric operators. However, we are interested in studying spin systems \textit{together with an embedding into an edge-restricted algebra} for some categorical symmetry. Here we give an abstract formulation.

\begin{defn}
Let $\mathcal{E}$ be an indecomposible multi-fusion category. A \textit{quotient} of $\mathcal{E}$ is an indecomposible multifusion category $\cD$ and a dominant tensor functor $F:\mathcal{E}\rightarrow \cD$
\end{defn}

We typically consider quotients up to (unitary) monoidal equivalence. Now, Given a dominant tensor functor $F:\mathcal{E}\rightarrow \mathcal{D}$, we have a natural inclusion of nets, 

$$A(\mathcal{E},X)\hookrightarrow_{\iota_{F}} A(\mathcal{D},F(X)) $$

\medskip

where $f\in A(\mathcal{E},X)_{I}=\mathcal{E}(X^{\otimes n},X^{\otimes n})$ is sent to the composition

$$
\begin{tikzcd}
F(X)^{\otimes n}\arrow[r, "can"] & F(X^{\otimes n}) \arrow[r, "F(f)"] & F(X^{\otimes n})\arrow[r, "can"] & F(X)^{\otimes n}\\
\end{tikzcd}
$$

in $A(\mathcal{D},F(X))_{I}$. Here $\text{can}$ denotes the canonical unitary isomorphisms built from the tensorators on $F$. $\iota_{F}$ is a unital inclusion of the local C*-algebras for each $I$, and is compatible with the connecting inclusions associated with $I\subseteq J$. Thus $\iota_{F}$ extends to a unital inclusion of quasi-local algebras. If we have a quotient  $F:\mathcal{C}\rightarrow \mathcal{D}$ and a quotient $G:\mathcal{D}\rightarrow \mathcal{E}$, then the composition $G\circ F:\mathcal{C}\rightarrow \mathcal{E}$ is a quotient. Furthermore, it is easy to see by construction that $\iota_{G\circ F}:= \iota_{G}\circ \iota_{F}$. We call such an inclusion of nets \textit{a categorical inclusion}.

Given a finitely semi-simple category $\mathcal{M}$, it's category of endofunctor $\text{End}(\mathcal{M})$ is an indecomposable multi-fusion category equivalent to $\text{Mat}_{n}(\text{Hilb}_{f.d.})$, where $n=\text{rank}(\mathcal{M})$. Recall that the data of a (left) $\mathcal{E}$-module category structure on a  finitely semi-simple category $\mathcal{M}$ can be expressed as a tensor functor $F:\mathcal{E}\rightarrow \text{End}(\mathcal{M})$, where $F(Y):=Y\triangleright \cdot$

\begin{defn}\label{spatialrealization} A \textit{spatial realization} of a fusion categorical net $A(\mathcal{E},X)$ is a categorical inclusion $A(\mathcal{E}, X)\hookrightarrow A(\text{End}(\mathcal{M}), X\triangleright \cdot)$ associated to an indecomposable left $\mathcal{E}$-module category $\mathcal{M}$. We denote $A(\text{End}(\mathcal{M}), X\triangleright \cdot)$ by $A(\mathcal{E}, X)_{\mathcal{M}}$.
\end{defn}

When $\text{rank}(\mathcal{M})=1$, $\text{End}(\mathcal{M})\cong \text{Hilb}_{f.d.}$, hence a spatial realization is an honest, unital inclusion of a fusion spin chain into the local operators of a concrete spin chain.

The motivation for this definition arises from our concept of fusion categorical symmetries. Above, we characterized these with the data of a fusion category $\mathcal{C}$, a right, indecomposable module category $\mathcal{M}$, and a (strongly tensor generating) object $X\in \mathcal{C}^{*}_{\mathcal{M}}$, correspond precisely (up to an obvious notion of equivalence) to spatial realizations of the fusion categorical net $A(\mathcal{C}^{*}_{\mathcal{M}},X)$.

If we change our perspective and denote $\mathcal{E}:=\mathcal{C}^{*}_{\mathcal{M}}$, then we can view $\mathcal{M}$ as a left indecomposable $\mathcal{E}$-module category, and thus spatial realizations of the fusion spin chain $A(\mathcal{E}, X)$, parameterized by a left module category $\mathcal{M}$, correspond precisely the algebra symmetry operators inside $A^{\circ}$ for an action of the dual category $\mathcal{C}^{*}_{\mathcal{M}}\curvearrowright A$. This leads us to reformulate Question \ref{dualityquest} as follows:

\bigskip
\noindent \textbf{Reformulation of Question \ref{Original question}}. 
    Let $A(\mathcal{E},X)$ and $ A(\mathcal{D},Y)$ be fusion spin chains, with spatial realizations parameterized by module categories $\mathcal{M}$ and $\mathcal{N}$ respectively. If $\alpha: A(\mathcal{E},X)\rightarrow A(\mathcal{D},Y)$ is a bounded-spread isomorphism, when does $\alpha$ extend to a bounded-spread isomorphism ${\widetilde{\alpha}: A(\mathcal{E},X)_{\mathcal{M}}\rightarrow A(\mathcal{D},Y)_{\mathcal{N}}}$?

\section{The algebra model for categorical inclusions}
\label{sec:4}

One way to build a categorical inclusion is to pick a connected, commutative Q-system object $L\in \mathcal{Z}(\mathcal{C})$ and consider the \emph{category $\mathcal{C}_{L}$ of right $A$-modules in $\mathcal{C}$} (technically we should first apply the forgetful functor to $L$ to obtain an object in $\mathcal{C}$). The commutative central structure on $L$ equips $\mathcal{C}_{L}$ with the structure of a unitary multi-fusion category such that the free module functor $F_{L}:\mathcal{C}\rightarrow \mathcal{C}_{L}$, $x\mapsto x\otimes L$, is a dominant tensor functor.

A theorem of~\cite{Natale2011} shows that any quotient $F:\cC\rightarrow \cD$ is of the form described above, i.e. there exists a (unique up to isomorphism) connected commutative Q-system $L\in \mathcal{Z}(\cC)$ such that the digram of monoidal functors commutes up to natural isomorphism

$$
\begin{tikzcd}
\mathcal{C}\arrow[d, swap, "F"]\arrow[dr,"F_{L}"] & \\
\mathcal{D} \arrow[r, swap,"\cong"] & \mathcal{C}_{L} \\
\end{tikzcd}
$$

This gives an ``internal" description of a quotient as opposed to the a-priori ``external" description requiring an indecomposable multi-fusion category $\mathcal{D}$. This is directly analogous to the first isomorphism theorem for groups.

Now we recall a concrete model for constructing $\mathcal{C}_{L}$ and the free module functor $F_{L}:\mathcal{C}\rightarrow \mathcal{C}_{L}$. Let $L\in \mathcal{Z}(\mathcal{C})$ be a connected, commutative $Q$-system. We denote the multiplication morphism $m:L\otimes L\rightarrow L$ and the unit map $\iota:\mathbb{1}\rightarrow L$. We slightly abuse notation and identify $L$ and its structure maps with their images in $\mathcal{C}$ under the forgetful functor.

Define the category $\mathcal{C}^{\circ}_{L}$ whose objects are objects of $\mathcal{C}$ and whose morphisms are given by 
$$\mathcal{C}^{\circ}_{L}(x,y):=\mathcal{C}(x,y\otimes A).$$
Composition of morphisms $f\in \mathcal{C}^{\circ}_{L}(y,z)$ and $g\in \mathcal{C}^{\circ}_{L}(x,y)$ is defined as 
$$f\circ g:= (1_{z}\otimes m)\circ (f\otimes 1_{L})\circ g\in \mathcal{C}^{\circ}_{L}(x,z).$$
Following \cite{MR3687214}, it is not hard to see that $\mathcal{C}^{\circ}_{L}$ becomes a C*-category with $*$-structure (denoted here with a $\dagger$ to distinguish from the $*$-structure in $\mathcal{C}$)
$$g^{\dagger}:=g^{*}\circ \left(1_{y}\otimes (m^{*}\circ \iota)\right).$$
To define the tensor structure on $\mathcal{C}^{\circ}_{L}$, the tensor product on objects (which we denote with $\boxtimes$) is just the tensor product in $\mathcal{C}$, $x\boxtimes y=x\otimes y$ while the tensor product of morphisms is defined, for $f\in \mathcal{C}^{\circ}_{L}(x,y)$ and $g\in \mathcal{C}^{\circ}_{L}(z,w)$ by
$$f\boxtimes g:=(1_{y\otimes w}\otimes m) \circ (1_{y}\otimes (\sigma_{L,w}\otimes 1_{L})\circ (f\otimes h). $$
This turns $\mathcal{C}_{L}$ into a (strict) unitary multi-tensor category, and since $L$ was connected in $\mathcal{Z}(\mathcal{C})$, it is in fact an indecomposable multi-fusion category. In fact, the functor $\mathcal{C}^{\circ}_{L}$ to the category $\text{FreeMod}_{L}$ of free (internal to $\mathcal{C}$) right $L$-modules sending $x\mapsto x\otimes L$ and $f\mapsto (1_{y}\otimes m)\circ (f\otimes 1_{L})$ is an equivalence of categories.

Thus the (unitary) Karoubian completion of $\mathcal{C}^{\circ}_{L}$ is equivalent to the category of projective right $L$-modules, which is equivalent to the category of all right $L$ modules since this category is semisimple. By a minor abuse of notation, we denote the unitary Karoubian completion of $\mathcal{C}^{\circ}_{L}$ by $\mathcal{C}_{L}$ and identify $\mathcal{C}^{\circ}_{L}$ with its image in $\mathcal{C}_{L}$.

\bigskip

\noindent In this picture, the free module functor $F_{L}:\mathcal{C}\rightarrow \mathcal{C}_{L}$ is just given by $F_{L}(x)=x\in \mathcal{C}_{L} $ and $F_{L}(f):=(1_{y}\otimes \iota)\circ f$.

Now we have a nice way to describe the net obtained from categorical quotients $F: \mathcal{C}\rightarrow \mathcal{D}$ directly in terms of $\mathcal{C}$. Let $L$ be the commutative Q-system in the center corresponding to this quotient, and $F_{L}$ the model for $F$ as above. Then for an interval $I\subseteq \mathbb{Z}$ with $|I|=n$, $$A(\mathcal{C}_{L}, X)_{I}=\mathcal{C}_{L}(X^{\otimes n}, X^{\otimes n})\cong \mathcal{C}(X^{\otimes n}, X^{\otimes n}\otimes L).$$

Furthermore, for $I\subseteq J$, the inclusion $A(\mathcal{C}_{L}, X)_{I}\hookrightarrow\fA(\mathcal{C}_{L}, X)_{J}$ is described for $f\in \mathcal{C}(X^{\otimes n}, X^{\otimes n}\otimes L)$ by
$$f\mapsto 1_{X^{\otimes l+n}}\sigma_{L, X^{\otimes r}}(1_{X^{\otimes l}}\otimes f\otimes 1_{X^{\otimes r}}).$$
\noindent Here $l$ is the number of points in $J$ strictly less than $I$, and $r$ is the number of points in $J$ strictly greater than $I$.

The inclusion $A(\mathcal{C}, X)_{I}\hookrightarrow A(\mathcal{C}_{L}, X)_{I}$ is simply given by 
$$f\mapsto f\otimes \iota,$$
\noindent for $f\in \mathcal{C}(X^{\otimes n}, X^{\otimes n})$.

This generalizes to give a clean description of the lattice of intermediate categorical inclusions. For any sub Q-system $K\le L$,  the inclusion map $j_{K}: K\rightarrow L$ gives the natural maps

$$A(\mathcal{C}_{K}, X)_{I}\cong \mathcal{C}(X^{\otimes n}, X^{\otimes n}\otimes K)\hookrightarrow \mathcal{C}(X^{\otimes n}, X^{\otimes n}\otimes L)\cong A(\mathcal{C}_{L}, X)_{I},$$

$$f\mapsto (1_{X^{\otimes n}}\otimes j_{K})\circ f.$$

\noindent This sets up an equivalence between the lattice of intermediate category inclusions between $A(\mathcal{C},X)$ and $A(\mathcal{C}_{L},X)$ and the lattice of sub Q-system objects $K\le L$. 
We typically take the largest algebra $L$ to be normalized so that $m\circ m^{*}=1_{L}$, but for intermediate subalgebras, we only require that $m\circ m^{*}=\lambda 1_{K}$ for some (necessarily positive) scalar $\lambda$.

\begin{rem}\label{Canonicalcond}\textbf{Canonical conditional expectation}. For any connected \'etale algebra $L\in \mathcal{Z}(\cC)$, identifying $A(\cC,X)\hookrightarrow A(\cC_{L},X)$, there is a canonical, locality-preserving conditional expectation $E:A(\cC_{L},X)\rightarrow A(\cC,X)$ which is defined for $f\in A(\cC_{L}, X)_{I}=\cC(X^{\otimes n}, X^{\otimes n}\otimes L)$ by
$$E(f):=\frac{1}{d_{L}} (1_{X^{\otimes n}}\otimes i^{*})\circ f.$$

\end{rem}

\begin{rem}
By the above discussion, spatial realizations of $A(\mathcal{C}, X)$ are parameterized by choices of indecomposable $\mathcal{C}$-module categories. 
Every such module category canonically defines a Lagrangian algebras $Z(\mathcal{M})\in \mathcal{Z}(\mathcal{C})$, its center \cite{MR2669355}. In particular, we can always view the spatial realizations $A(\mathcal{C},X)\hookrightarrow A(\mathcal{C},X)_{\mathcal{M}}$ as $A(\mathcal{C}, X)\hookrightarrow A(\mathcal{C}_{Z(\mathcal{M}}, X)$.
\end{rem}

In this section, we prove a lemma necessary for the proof of the main theorem, concerning relative commutants.  As mentioned above, fusion spin chains satisfy weak algebraic Haag duality, and in fact, satisfy an even stronger version simply called algebraic Haag duality. In particular, if we choose $n$ such that $X^{\otimes n}$ contains a copy of every simple object in $\cC$ (the existence of such an $n$ is the definition of strong tensor generating), then we have that for any interval $I$ with $|I|>n$, $A(\cC,X)_{I}=\{x\in \fA\ :\ [x,y]=0\ \text{for all}\ y\in A(\cC,X)_{I^{c}}\}$. The following lemma shows that the nets $\fA(\cC_{L},X)$ satisfy an even stronger ``relative" version of this property.

\begin{lem}\label{Haagdualitylem}
Let $A(\cC, X)$ be a fusion spin chain and $n$ as above. For any connected, commutative Q-system $L\in \mathcal{Z}(\cC)$, then for any interval $I$ with $|I|>n$, $$A(\cC_{L},X)_{I}=\{x\in A(\cC_{L}, X)\ :\ [x,y]=0\ \text{for all}\ y\in A(\cC,X)_{I^{c}}\}.$$
\end{lem}

\begin{proof}
Note that we can view the quasi-local algebra $A_{\cC_{L}}$ for $A(\cC,X)$ as a finite index extension of $A_{I^{c}}$, supported on the canonical copy of $\cC\boxtimes \cC^{mp}\hookrightarrow \text{Bim}(A_{I^{c}})$ \cite{2304.00068}. Indeed, $A(\cC_{L}, X)$ is the Q-system realization of the Q-system $K\in \cC\boxtimes \cC^{mp}$ obtained from considering $\cC_{L}$ as a $\cC$-$\cC$ bimodule category,  choosing the object $X^{\otimes n}\in \cC_{L}$ and taking internal endomorphisms. In symbols,
$$A(\cC_{L}, X)\cong |K|$$
Since the embedding of $\cC\boxtimes \cC^{mp}\hookrightarrow \text{Bim}(\fA_{I^{c}})$ is fully faithful, by the usual Q-system arguments from subfactor theory \cite{2304.00068, MR4419534}, 
$$A(\cC,X)_{I^{c}}^{\prime}\cap A(\cC_{L},X) \cong 1_{A(\cC,X)_{I^{c}}}\otimes_{\mathbbm{C}} \cC(\mathbbm{1},K)\subseteq |K|.$$

By by the definition of internal end and Q-system realization, under the identification of the realization subspace with our concrete net, we obtain 
$$A(\cC_{L},X) \cong 1_{A(\cC,X)_{I^{c}}}=\cC(X^{\otimes n}, X^{\otimes n}\otimes L)=A(\cC_L, X)_{I}.$$
\end{proof}

\section{Solving the extension problem with DHR bimodules}
\label{sec:5}

Here, we recall the definition of a DHR bimodule on a net of algebras (over $\mathbb{Z}$). First, if $A$ is a (unital) C*-algebra, then a \textit{bimodule} is an algebraic $A$-$A$ bimodule $X$, together with a right $A$-valued inner product $\langle \cdot\ |\ \cdot \rangle: X\times X\rightarrow A$ (making $X$ into a Hilbert A-module) such that the left $A$ action is by adjointable operators. What we call bimodules are more typically called \textit{correspondences} in the C*-algebra literature.

We denote the collection of all $A$-$A$ bimodules by $\text{Bim}(A)$. This forms a C*-tensor category, whose objects are bimodules, morphisms are adjointable bimodule intertwiners, and tensor product is given by the relative tensor product of bimodules, denoted $\boxtimes_{A}$ (see \cite{MR4419534}).

If $X\in \text{Bim}(A)$, a \textit{projective basis} is a finite set $\{b_{i}\}\subseteq X$ such that for all $x\in X$ 

$$x=\sum_{i} b_{i}\langle b_{i} | x\rangle.$$

\begin{defn} Let $A$ be an abstract spin chain with quasi-local algebra $A$. A bimodule over the quasilocal algebra $A$ is called a \textit{DHR bimodule} if there exists an $R$ such that for every interval with $|I|>R$, there exists a projective basis $\{b_i\}$ for $X$ such that $ab_{i}=b_{i}a$ for all $x\in A_{I^{c}}$.
\end{defn}

The collection of DHR bimodules is closed under $\boxtimes_{A}$ and thus defines a full $\otimes$ subcategory of $\text{DHR}(A)\subseteq \text{Bim}(A)$. If $A$ also satisfies \textit{weak algebraic Haag duality} (a property satisfied by fusion spin chains, see \cite{2304.00068}),
then $\text{DHR}(\fA)$ admits a \textit{canonical braiding}. 

The clearest physical interpretation of DHR bimodules as a braided tensor category is in terms of bulk patch operators. One realization of a fusion spin chain is as the ``boundary algebra" of of 2+1D topologically ordered spin system, in particular of Levin-Wen type models~\cite{Levin2005}. In this case, string operators associated to an anyon type $Z\in \mathcal{Z}(\mathcal{C})$ in the bulk can be pushed onto the boundary, giving operators in the fusion spin chain $A(\mathcal{C},X)$ called \textit{patch operators} \cite{PhysRevB.107.155136, inamura202321dsymmetrytopologicalorderlocalsymmetric}. 
If we push the right endpoint off to infinity, this operator no longer lives in $A(\mathcal{C},X)$, but can be thought of formally as an element in a DHR-bimodule since it is localized in an interval around the left endpoint. Similarly to the GNS construction, we can take local perturbations and assemble this into a DHR-bimodule, corresponding to the anyon type $Z\in \mathcal{Z}(\mathcal{C})$. This suggests there should be a close relationship between between $\text{DHR}(A(\mathcal{C}, X))$ and $\mathcal{Z}(\mathcal{C})$, which is demonstrated in the following theorem.

\begin{thm}[\cite{2304.00068}, Theorem C]
    If $X$ is a strong, self-dual tensor generator of the unitary fusion category $\mathcal{C}$, then $\text{DHR}(\fA(\mathcal{C},X))\cong \mathcal{Z}(\mathcal{C})$ as braided  C*-tensor categories.
\end{thm}

The above theorem provides an MPO picture for DHR bimodules, in the case of a spatial realization. First, for an object $Z\in \mathcal{Z}(\mathcal{C})$, we consider defect tensors $a$ as in the following picture, where the horizontal string is labeled by an object $Z$ in the Drinfeld center, and $M$ is the associated MPO for the object $Z\in \mathcal{C}$. 
\begin{align}
    \vcenter{\hbox{\includegraphics[page=4]{Figures.pdf}}}
\end{align}
See \cite{MR3614057,Williamson2017SET} for a description of how to decompose the space of defect tensors for MPOs corresponding to a category $\mathcal{C}$ into simple objects of the center $\mathcal{Z}(\mathcal{C})$. 

Next, we can represent the half-braiding associated to $Z$ with an MPO $h$, and in this way extend the action of the fusion ring by channels on local operators to the defect tensors via the following picture. 
\begin{align}
    \vcenter{\hbox{\includegraphics[page=5]{Figures.pdf}}}
\end{align}
The defect sectors are decomposed into irreducible representations under this action, which is similar to Ocneanu's tube algebra \cite{MR0996454}. 

Finally, we consider the vector space of all such defect tensors invariant under the $\mathcal{C}$-symmetry. This is a bimodule under the action of symmetric local operators $L_1,L_2,\dots$ and $R_1,R_2,\dots$ as depicted in the following picture.
\begin{align}
    \vcenter{\hbox{\includegraphics[page=6]{Figures.pdf}}}
\end{align}
By the above theorem, it turns out that these constitute \textit{all} the abstractly defined DHR bimodules in the symmetric case.

One of the main motivations of \cite{2304.00068} 
is that the DHR construction defines a functor $\text{DHR}:\text{Net}_{\mathbbm{Z}}\rightarrow \text{C*-BrTens}$, where $\text{Net}_{\mathbbm{Z}}$ is the groupoid of abstract spin chains and bounded-spread isomorphisms. This associates to each abstract spin chain its braided tensor category of DHR bimodules, while
$$\text{DHR}(\alpha): \text{DHR}(A)\rightarrow \text{DHR}(B)$$
which acts on DHR bimodules by sending a DHR bimodule $X$ of $A$ to $\text{DHR}(\alpha)(X):=X$ as a vector space, with left and right $B$ actions given by
$$b\triangleright x\triangleleft c:=\alpha^{-1}(b)x\alpha^{-1}(c),$$
\noindent with $B$ valued inner product
$$\langle x\ |\ y\rangle_{B}:=\alpha(\langle x\ |\ y\rangle_{A}).$$

\medskip

In particular, for every bounded-spread isomorphism between fusion spin chains 
$$\alpha: A(\mathcal{C},X)\rightarrow A(\mathcal{D},Y),$$
\noindent there is an induced unitary braided equivalence of categories 
$$\text{DHR}(\alpha): \mathcal{Z}(\cC)\cong \text{DHR}( A(\mathcal{C},X))\rightarrow A(\mathcal{D}, Y)\cong \mathcal{Z}(\mathcal{D}).$$

\bigskip

\begin{thm}\label{mainthmalt}
    Let $\alpha: A(\mathcal{C}, X)\rightarrow A(\mathcal{D}, Y)$ be a bounded-spread isomorphism. Then for any commutative, connected Q-systems $L\in \mathcal{Z}(\mathcal{C})$ and $K\in\mathcal{Z}(\mathcal{D})$, extensions of $\alpha$ to bounded-spread isomorphisms $\widetilde{\alpha}:A(\mathcal{C}_{L}, X)\rightarrow A(\mathcal{D}_{K}, Y)$ are in bijective correspondence with algebra isomorphisms $\gamma:\text{DHR}(\alpha)(L)\cong K$ in $\mathcal{Z}(\mathcal{D})$. As a consequence:
    
    \begin{enumerate}
    \item 
    If $\text{DHR}(\alpha)(L)$ is not isomorphic to $K$, then $\alpha$ has no extension to a bounded-spread isomorphism $\widetilde{\alpha}:A(\mathcal{C}_{L}, X)\rightarrow A(\mathcal{D}_{K}, Y)$

    \medskip
    
    \item 
    If $\text{DHR}(\alpha)(L)\cong K$, spatial implementations form a torsor over $\text{Aut}(L)$.
    \end{enumerate}
\end{thm}

\begin{proof} Let $L\in \mathcal{Z}(\cC)$ be a commutative, connected Q-system. Then $A(\cC_{L}, X)$ is algebraically a bimodule over the quasi-local algebra $A(\cC, X)$ obtained by including the latter as a subalgebra of the former. The canonical condition expectation equips $A(\cC_{L}, X)$ with the structure of a right Hilbert $A(\cC, X)$ module, hence $(A(\cC_{L}, X),E)\in \text{Bim}(A(\cC, X))$. In fact, from the canonical embedding $\mathcal{Z}(\cC)\cong \text{DHR}(A(\cC,X))\subseteq \text{Bim}(A(\cC,X))$, we have that the quasi-local algebra with condition expectation $A(\cC_{L}, X),E)$ is simply the Q-system realization $|L|$ of $L\in \mathcal{Z}(\mathcal{C})$, and in particular $(A(\cC_{L}, X),E)\in \text{DHR}(A(\cC, X))$ is a connected, commutative Q-system.

With this picture in hand, to establish our bijection, we simply show that the data of one thing can simply be reinterpreted as the data of the other. 

First, suppose $\widetilde{\alpha}: A(\mathcal{C}_{L}, X)\rightarrow A(\mathcal{D}_{K}, Y) $ is a bounded-spread isomorphism, restricting to $\alpha$ on the subsystem $A(\mathcal{C}, X)$. We claim that $\widetilde{\alpha}$ \textit{is} an algebra object isomorphism $\text{DHR}(\alpha)(A(\mathcal{C}_{L}, X)) \cong A(\mathcal{D}_{K}, X)$. Indeed, $\text{DHR}(\alpha)(A(\mathcal{C}_{L}, X))$ is simply $A(\mathcal{C}_{K}, X)$ with left and right action twisted by $\alpha^{-1}$. We check for $a\in A(\mathcal{C}_{L}, X)$ and $b,c\in A(\mathcal{D}, Y)$ 
$$\widetilde{\alpha}(b\triangleright a\triangleleft c)=\widetilde{\alpha}(\alpha^{-1}(b)a\alpha^{-1}(c))=b\widetilde{\alpha}(a) c$$
\noindent where we have used $\widetilde{\alpha}|_{A(\mathcal{C}, X)}=\alpha$. Thus $\widetilde{\alpha}:\text{DHR}(\alpha)(A(\mathcal{C}_{L}, X))\cong A(\mathcal{D}_{K}, X)$ is an isomorphism of $A(\mathcal{D}, Y)$ bimodules. Since it is an algebra isomorphism, this bimodule isomorphism is an isomorphism of algebra objects.

Conversely, suppose we have an isomorphism of algebra objects $\widetilde{\alpha}:\text{DHR}(\alpha)(A(\mathcal{C}_{L}, X))\cong A(\mathcal{D}_{K}, X) $ in the category of $A(\mathcal{D}, Y)$ bimodules. Then reversing the above logic, $\widetilde{\alpha}$ is an algebra isomorphism $A(\mathcal{C}_{L}, X)\rightarrow A(\mathcal{D}_{K}, Y)$ which restricts to $\alpha$ on the algebra $A(\mathcal{C}, X)$. It remains to show that $\widetilde{\alpha}$ is bounded spread.

Suppose $\alpha$ has spread at most $R$. We claim $\widetilde{\alpha}$ has spread at most $R$. Indeed, let $a\in A(\mathcal{C}_{L}, X)_{I}$. Note that $A(\mathcal{D},Y)_{(I^{+R})^{c}}\subseteq \alpha(A(\mathcal{C},X)_{I^{c}})$, and since $[a,b]=0$ for all $b\in A(\mathcal{C},X)_{I^{c}}$, 
$$0=[\widetilde{\alpha}(a), \widetilde{\alpha}(b)]=[\widetilde{\alpha}(a), \alpha(b)],$$
and hence 
$$[\widetilde{\alpha}(A(\mathcal{C}_{L}, X)_{I}), A(\mathcal{D}, Y)_{(I^{+R})^{c}}]=0.$$
\noindent By Lemma \ref{Haagdualitylem}, this implies
$$\widetilde{\alpha}(A(\mathcal{C}_{L}, X)_{I})\subseteq A(\mathcal{C}_{K}, X)_{I^+R},$$
\noindent so that $\widetilde{\alpha}$ is bounded spread. This concludes the proof of the main statement of the theorem.

The first consequence listed above follows  immediately. The second follows from the fact that for any isomorphisms $\gamma, \gamma^{\prime}: \text{DHR}(\alpha(L))\cong K$,  $\gamma^{\-1}\circ\gamma^{\prime}\in \text{Aut}(\text{DHR}(\alpha(L))\cong \text{Aut}(L)$.

\end{proof}.

We recommend the reader to compare the above theorem with an analogous result concerning extending symmetries in 1+1D algebraic quantum field theory, \cite[Section 6]{Bischoff2019}. 

\bigskip

\begin{proof}[\unskip\nopunct]\textbf{Proof of Theorem \ref{mainthm}}. The Theorem \ref{mainthm} stated in the introduction follows as a special case of the above, and we believe it provides a satisfactory answer to Question \ref{dualityquest}. Indeed, consider the case where $L \in \mathcal{Z}(\mathcal{C})$ and $K\in \mathcal{Z}(\mathcal{D})$ are $Z(\mathcal{M})$ and $Z(\mathcal{N})$ associated to the $\mathcal{C}$ and $\mathcal{D}$ module categories, respectively. Then we have identifications of the spatial realizations $A(\mathcal{C},X)_{\mathcal{M}}\cong A(\mathcal{C}_{Z(\mathcal{M})}, X)$, whence the first part of the result follows from the previous theorem. That $\text{Aut}(Z(\mathcal{M}))\cong (\mathcal{C}^{*}_{\mathcal{M}})^{\times}$ follows, for example, from \cite{MR4357481}. 
\end{proof}
\bigskip

\noindent For any algebra $L\in \mathcal{Z}(\mathcal{C})$, there is a group
$$\text{Aut}_{br}(\mathcal{Z}(\cC)\ |\ L)$$
\noindent which consists of equivalence classes of pairs $(\alpha, \gamma)$, where $\alpha$ is a braided tensor equivalence of $\mathcal{Z}(\cC)$ and $\gamma: \alpha(L)\cong L$ is an isomorphism of algebras. An equivalence between pairs $(\alpha, \gamma)$ and $(\beta, \delta)$ consists of a monoidal natural isomorphism $\nu: \alpha\cong \beta$ such that $\delta\circ \nu_{L}=\gamma$. By \cite{Davydov2013a,Schatz24}, 
for a Lagrangian algebra $L$ corresponding to an indecomposable $\mathcal{C}$-module category $\mathcal{M}$, we have
$$\text{Aut}_{br}(\mathcal{Z}(\cC)\ |\ L)\cong \text{Aut}_{\otimes}(\mathcal{C}^{*}_{\mathcal{M}})$$
where the latter is the group of tensor autoequivalences up to monoidal natural isomorphism.

Let $\text{QCA}(\mathcal{C},X,\mathcal{M})$ be the group of bounded-spread autoequivalences of $A(\mathcal{C},X)_{\mathcal{M}}$ which preserve the subalgebra $A(\mathcal{C}, X)$. We can think of this as the group of spatially implemented dualities, or as the group of ``edge-restricted" QCA that preserve the symmetric subalgebra. Then the above theorem yields the following corollary:

\begin{cor}
There is a homomorphism  $\pi:\text{QCA}(\mathcal{C},X,\mathcal{M})\rightarrow \text{Aut}_{\otimes}(\mathcal{C}^{*}_{\mathcal{M}})$.
\end{cor}

We can easily choose examples so that this is surjective. Indeed, consider the object $X=\bigoplus_{Y\in \text{Irr}(\mathcal{C})} Y$. Then $\pi: \text{QCA}(\mathcal{C}, X, \mathcal{C})\rightarrow \text{Aut}_{\otimes}(\mathcal{C})$ is surjective (see \cite{2304.00068}).

A family of examples of interest, which arise from onsite finite group symmetry, are when $\mathcal{C}=\text{Rep}(G)$, $\mathcal{M}=\text{Hilb}_{f.d.}$, made into a $\mathcal{C}$-module category with the standard fiber functor structure on $\text{Rep}(G)$. In this case, we obtain a homomorphism 
$$\pi: \text{QCA}(\text{Rep}(G),X,\text{Hilb}_{f.d.})\rightarrow \text{Aut}_{\otimes}(\text{Hilb}_{f.d.}(G))\cong H^{2}(G, \text{U}(1))\rtimes \text{Aut}(G).$$

\bigskip

\section{Examples and Applications}
\label{sec:6}

In this section, we consider a family of examples generalizing Kramers-Wannier (KW) dualities \cite{PRXQuantum.4.020357}, and we also discuss applications of our previous result.

\bigskip

\noindent \textbf{Generalized shifts.} Let $\mathcal{C}$ be a unitary fusion category, and $\alpha\in \text{Aut}_{br}(\mathcal{Z}(\mathcal{C}))$, such that $\alpha$ generates a categorical action of $\mathbbm{Z}/n\mathbbm{Z}$ on $\mathcal{Z}(\mathcal{C})$ for some $n$. Then the $o_{4}$ obstruction necessarily vanishes as $H^{4}(\mathbbm{Z}/n\mathbbm{Z}, U(1))$ is trivial, and thus by \cite{MR2677836} there exists a $\mathbbm{Z}/n\mathbbm{Z}$-graded extension 

$$\mathcal{D}:=\bigoplus_{g\in \mathbbm{Z}/n\mathbbm{Z}} \mathcal{C}_{g},$$

\noindent with $\mathcal{C}_{e}=\mathcal{C}$.
For any $Y\in \mathcal{C}_{\alpha}$, $Y^{\otimes n}\in \mathcal{C}_{e}=\mathcal{C}$. We can always choose $Y$ such that $X:=Y^{\otimes n}$ is a strong tensor generator for $\mathcal{C}$, for example, $Y=\bigoplus_{Z\in \text{Irr}(\mathcal{C}_{\alpha})} Z$. 

Now, consider the fusion spin chain $A(\mathcal{C},X)$. 
We define a ``generalized translation" QCA for $a\in A(\mathcal{C},X)_{[i,j]}\cong \cC(X^{j-i+1}, X^{j-i+1})\cong \cD(Y^{n(j-i+1)},Y^{n(j-i+1)})$,
$$\tau_{\alpha}(a):=1_{Y^{n-1}}\otimes a\otimes 1_{Y}\in A(\mathcal{C},X)_{[i-1,j]}.$$
Then since for any $Z\in \mathcal{Z}(\mathcal{C})$ we have natural isomorphisms $Y\otimes Z\cong \alpha(Z)\otimes Y$, it is easy to verify that 
$$\text{DHR}(\tau_{\alpha})\cong \alpha.$$
Thus, any braided autoequivalence of $\mathcal{Z}(\cC)$ can be implemented by a QCA on the \text{some} fusion spin chain of the form $A(\mathcal{C}, X)$, although the choice of $X$ depends on the autoequivalence.

\bigskip

\noindent \textbf{Dualities from Tambara-Yamagami}. The easiest example is the case $\mathcal{C}:=\text{Hilb}_{f.d.}(A)$, where $A$ is a finite abelian group. $\mathcal{Z}(\mathcal{C})\cong \text{Hilb}_{f.d.}(A\times \widehat{A})$. Any symmetric, non-degenerate bicharacter $\chi:A\times A\rightarrow U(1)$, yields an isomorphism $\widetilde{\chi}:A\rightarrow \widehat{A}$, given by $\widetilde{\chi}(x)=\chi(a, \cdot)$. This results in a braided autoequivalence $\alpha_{\chi}\in \text{Aut}_{br}(\text{Hilb}_{f.d.}(A\times \widehat{A}))$ defined by
$$\alpha_{\chi}(a, f)=(\widetilde{\chi}^{-1}(f),\widetilde{\chi}(a)).$$
Which further results in a $\mathbbm{Z}/2\mathbbm{Z}$ extension of $\text{Hilb}_{f.d.}(A)$ called a \textit{Tambara-Yamagami} category \cite{TAMBARA1998692}, whose simple objects are simply $A\cup \{\rho\}$ with fusion rules 
$$a\otimes b=ab\, ,$$
$$a\otimes \rho=\rho=\rho\otimes a\, ,$$
$$\rho\otimes \rho\cong \bigoplus_{a\in A} a \, .$$

If we set $Y=\rho$, then $X:=Y^{2}=\bigoplus_{a\in A} a$ is a strong tensor generator for $\mathcal{C}$, and thus the above construction yields a QCA on the spin chain
$$\tau_{\alpha_\chi}: A(\mathcal{C}, X)\rightarrow A(\mathcal{C}, X)$$
\noindent with $\text{DHR}(\tau_{\alpha_\chi})\cong \alpha_{\chi}$. For $A=\mathbbm{Z}_{2}$, there is a unique symmetric, non-degenerate bicharacter $\chi$, and $\tau_{\alpha_{\chi}}$ is exactly the Kramers-Wannier duality map described earlier.

There are two canonical spatial realizations of $A(\mathcal{C}, X)$ described by the Lagrangian algebras $L_{1}:=\mathbbm{C}[A]\times 1$ and $L_{2}:=1\times \mathbbm{C}[\widehat{A}]$, sometimes called electric and magnetic algebras. Then $\alpha_{\chi}$ swaps $L_{1}$ and $L_{2}$. In particular, by Theorem~\ref{mainthm},
$\tau_{\alpha_{\chi}}$ cannot extend to a QCA on either spatial realization $A(\mathcal{C}_{L_1}, X)$ or $A(\mathcal{C}_{L_2}, X)$.

\bigskip

\noindent \textbf{Q-system complete fusion categories}. A unitary fusion category $\mathcal{C}$ is called \textit{Q-system complete} (also called \textit{torsion-free in the literature} if there is exactly one indecomposable module category (namely, $\mathcal{C}$ itself) up to equivalence. This terminology comes from the fact that the fusion category $\mathcal{C}$, thought of as a 2-category with one object, is Q-system complete in the sense of \cite{MR4419534}. 
Examples include $\text{Fib}$ and more generally $\text{PSU}(2)_{2k+1}$, and $\text{Hilb}_{f.d.}(\mathbbm{Z}/p\mathbbm{Z}, \omega)$, where $[\omega]\in H^{3}(\mathbbm{Z}/p\mathbbm{Z})$ is non-trivial.

For Q-system complete fusion categories $\mathcal{C}$, any fusion spin chain of the form $A(\mathcal{C}, X)$ has  only one spatial realization $A(\mathcal{C}, X)_{\mathcal{C}}$, which in turn corresponds to a unique isomorphism class of Lagrangian algebra in $\mathcal{Z}(\mathcal{C})$. This results in the following corollary of Theorem \ref{mainthm}.

\begin{prop} Let $\mathcal{C}$ be a Q-system complete fusion category. Suppose $\alpha: A(\mathcal{C}, X)\rightarrow A(\mathcal{D}, Y)$ is a bounded-spread isomorphism between fusion spin chains. Then

\begin{enumerate}
\item 
$\mathcal{C}\cong \mathcal{D}$ as fusion categories and in particular $\mathcal{D}$ is Q-system complete.

\medskip

\item 
There exists a spatial implementation  $\widetilde{\alpha}:A(\mathcal{C}, X)_{\cC}\rightarrow A(\mathcal{D}, Y)_{\cD}$ of $\alpha.$

\medskip

\item 
The spatial implementations form a torsor over $\text{Inv}(\cC)$. In particular, if $\mathcal{C}$ is has no non-trivial invertible objects (e.g $PSU(2)_{2k+1}$), then there is a unique spatial implementation.
\end{enumerate}
\end{prop}

\begin{proof}
    Let $L$ denote the unique Lagrangian algebra object in $\text{DHR}( A(\mathcal{C}, X))\cong \mathcal{Z}(\cC)$, so that $\mathcal{C}\cong \mathcal{Z}(\cC)_{L}$. Then $\text{DHR}(\alpha)(L)\in \text{DHR}(A(\mathcal{D}, Y))\cong \mathcal{Z}(\cD)$ is the unique Lagrangian algebra (since $\text{DHR}(\alpha)$ is a braided equivalence). In particular, we have $K:=\text{DHR}(L)$ must be the canonical Lagrangian algebra in $\mathcal{Z}(\cD)$, hence $\mathcal{Z}(\cD)_{K}\cong \mathcal{D}$, and thus $\text{DHR}(\alpha)$ induces a monoidal equivalence  $\mathcal{C}\cong \mathcal{Z}(\cC)_{L}\rightarrow \mathcal{Z}(\mathcal{D})_{K}\cong \mathcal{D}$.

    Since the spatial realizations must correspond to the Lagrangian algebras $L$ and $K$ respectively, the existence of a spatial implementation of $\alpha$, and the statement about the torsor, follows immediately from Theorem \ref{mainthm}.
\end{proof}

\bigskip

\noindent \textbf{Classification of QCA.} In this section, apply our above analysis to give a better understanding of QCA which commute with a group symmetry, and connect our perspective to the literature. Suppose we have a group $G$ acting unitarily on the local Hilbert space $V:=\mathbbm{C}^{d}$, given by an assignment $g\mapsto U_{g}\in U(d)$. If we let $A$ denote the spin system  This extends to the automorphism of each $A_{I}$ by $\alpha_{g}(a):=(U_{g})^{\otimes n} a (U^{*}_{g})^{\otimes n}$, which extends to a homomorphism $G\rightarrow \text{Aut}(A)$. We assume this homomorphism is injective.

In the literature, there have been several investigations into \textit{symmetric QCA (sQCA)} \cite{PhysRevLett.124.100402}, which in our language are bounded-spread isomorphisms $\alpha: A\rightarrow A$ such that ${\alpha\circ \alpha_{g}=\alpha_{g}\circ \alpha}$ for all $g\in G$. Our goal in this section is to make contact with these results from our perspective.

First, we recast this picture into the language of fusion spin chains by noting that $A^{G}$ has the structure of an abstract spin chain with local algebras $A^{G}_{I}:=\{x\in A_{I}\ : \alpha_{g}(x)=x\ \text{for all}\ g\in G\}$. 

Now if we view $V$ as an object in $\text{Rep}(G)$, and we further assume that $V$ is self-dual and strongly tensor generating (e.g. $V\cong L^{2}(G)$ is the left regular representation), then 
$$A^{G}\cong A(\text{Rep}(G), V).$$

The inclusion $A^{G}\subseteq A$ is a spatial realization which corresponds to the standard fiber functor $\text{Rep}(G)\rightarrow \text{Hilb}_{f.d.}$. This leads to the following natural question:

\begin{question}
    Suppose $\alpha: A^{G}\rightarrow A^{G}$ is a duality and $\widetilde{\alpha}:A\rightarrow A$ is a spatial implementation. Under what conditions the spatial implementation $\widetilde{\alpha}$ a symmetric QCA?
\end{question}

To answer this question, note that the Lagrangian algebra resulting from this spatial realization in $\mathcal{Z}(\text{Rep}(G))$ can actually be viewed as an algebra object in $\text{Rep}(G)\subseteq \mathcal{Z}(\text{Rep}(G))$, where we identify $\text{Rep}(G)$ with the full subcategory of its center using the standard symmetric braiding. The algebra is simply $L:=\text{Fun}(G)$, the commutative algebra of $\mathbbm{C}$-valued functions, viewed as a $G$-representation where $G$ acts by translation~\cite{MR3242743}.  

\begin{thm}\label{thm:sQCA}
    Suppose $\alpha: A^{G}\rightarrow A^{G}$ is a duality and $\widetilde{\alpha}:A\rightarrow A$ is a spatial implementation, and $\gamma: \text{DHR}(\alpha)(L)\cong L$ is the associated algebra isomorphism. Then
    
    \begin{enumerate}
    \item     
    $\widetilde{\alpha}$ is a symmetric QCA if and only if for every $\beta\in \text{Aut}(L)$,  $\gamma\circ \text{DHR}(\alpha)(\beta)=\beta \circ \gamma$.

    \medskip
    
    \item 
    There is a homomorphism from the group of sQCA on $A$ to $H^{2}(G, \text{U}(1))$.
    \end{enumerate}
\end{thm}

\begin{proof} $A$, as an algebra object (Q-system) in $\text{Bim}(A^{G})$, is isomorphic to the Q-system realization $|L|$ of the function algebra $\text{Fun}(G)\in \text{Rep}(G)\subseteq \mathcal{Z}(\text{Rep}(G))$, and thus we have that the group $\text{Aut}(A|A^{G})\le \text{Aut}_{\text{Rep}(G)}(A)$ of automorphisms of $A$ that are the identity on $A^{G}$, is isomorphic to $\text{Aut}_{\text{Rep}(G)}(L)=G$. But the group action $g\mapsto \alpha_{g}$ is an injective homomorphism $G\rightarrow \text{Aut}(A|A^{G})\le G$, which therefore must also be surjective and we must have that $\text{Aut}(A|A^{G})=G$ since $G$ is finite. 

Furthermore, since $\text{Hilb}_{f.d.}(G)$ is the dual category of $\text{Rep}(G)$ associated to the Lagrangian $L$, we have that the linear span $\mathbbm{C}[\text{Aut}_{\text{Rep}(G)}(L)]=\text{End}_{\text{Rep}(G)}(L)$, by \cite{MR4357481} 
This proves part $(1)$. 
Part (2) follows from \cite{Schatz24}.
\end{proof}

\bigskip

Our next goal is to prove Corollary \ref{classificationcor} from the introduction. This requires us to review the index for categorical dualities of fusion spin chains, introduced in \cite{JL24}. Given an inclusion of C*-algebras $A\subseteq B$ with unique tracial state $\tau$, we can build an associated inclusion of von Neumann $\rm{II}_{1}$ factors $N=A^{\prime \prime}\subseteq M=B^{\prime \prime}$, where the completions are taken in the Hilbert space representation $L^{2}(A, \tau)$. Then the Jones index $[M:N]$ is a real number that satisfies several nice properties \cite{MR0696688}. Since $A,B$ are assumed to have a unique trace, their $\rm{II}_{1}$ factor completion is unambiguous, and thus we can slightly abuse notation and write $[B:A]:=[M:N]$.

Let $\alpha$ be a bounded spread automorphism $\alpha\in \text{Aut}(A)$. For any $z\in \mathbbm{Z}$, define $A_{z}\subseteq A$ as the C*-algebra of observables localized in the interval $(-\infty, z]$. Fusion spin chains $A:=A(\mathcal{C}, X)$ satisfy the unique trace property, as do all the algebras $A_{z}$.

Now choose  $x,y\in \mathbbm{Z}$ so that $y\le x$, $A_{y}\subseteq \alpha(A_{x})$, which can be accomplished due to the bounded spread condition.

Now we can define the index

$$\text{Ind}(\alpha):= \left(\frac{[\alpha(A_{x}):A_{y}]}{[A_{x}:A_{y}]}\right)^{\frac{1}{2}}.$$

For ordinary spin chains, this agrees with the GNVW index \cite[Proposition 4.2]{JL24}. Our goal is to show that for the group case, if a $G$-duality $\alpha$ duality admits an extension to a QCA $\widetilde{\alpha}$, then $\text{Ind}(\alpha)=\text{Ind}(\widetilde{\alpha})$. This would imply that if $\text{Ind}(\alpha)=1$, $\widetilde{\alpha}$ is a finite-depth circuit. 

\begin{prop}\label{indexlemma}
Let $\widetilde{\alpha}:A\rightarrow A$ be a QCA such that $\widetilde{\alpha}(A^{G})=A^{G}$, and let $\alpha:=\widetilde{\alpha}|_{A^{G}}$ Then $\text{Ind}(\widetilde{\alpha})=\text{Ind}(\alpha)$.
\end{prop}

\begin{proof}

Let $x\le y$ be as in the definition of index. Since the $G$-representation is a strong tensor generator, it follows that $A_{x}, A_{y}, A^{G}_{x}, A^{G}_{y}, \widetilde{\alpha}(A_{x}), \alpha(A^{G}_{x})$ are are simple C*-algebras with unique trace. Then we have the inclusions 

$$
\begin{tikzcd}
    A^{G}_{x}\arrow[r,symbol=\subset]&A_{x}\\
    A^{G}_{y}\arrow[r,symbol=\subset]\arrow[u,symbol=\subset]&A_{y}\arrow[u,symbol=\subset]
\end{tikzcd}\ \ \text{and}\ \ 
\begin{tikzcd}
    \alpha(A^{G}_{x})\arrow[r,symbol=\subset]& \widetilde{\alpha}(A_{x})\\
    A^{G}_{y}\arrow[r,symbol=\subset]\arrow[u,symbol=\subset]&A_{y}\arrow[u,symbol=\subset]
\end{tikzcd}
$$

\noindent Since the indices are actually indices of subfactors, the index is multiplicative \cite{MR0696688}, hence 

$$[A_{x}:A^{G}_{x}][A^{G}_{x}:A^{G}_{y}]=[A_{x}:A_{y}][A_{y}:A^{G}_{y}]$$

and

$$[\widetilde{\alpha}(A_{x}):\alpha(A^{G}_{x})][\alpha(A^{G}_{x}):A^{G}_{y}]=[\widetilde{\alpha}(A_{x}):A_{y}][A_{y}:A^{G}_{y}].$$

\bigskip

\noindent But from standard results in subfactor theory (for example \cite{MR1473221, MR1642584}), $$[A_{x}:A^{G}_{x}]=[A_{y}:A^{G}_{y}]=[\widetilde{\alpha}(A_{x}):\alpha(A^{G}_{x})]=[A_{y}:A^{G}_{y}]=|G|.$$

\noindent Thus we have $[A^{G}_{x}:A^{G}_{y}]=[A_{x}:A_{y}]$ and $[\alpha(A^{G}_{x}):A^{G}_{y}]=[\widetilde{\alpha}(A_{x}):A_{y}]$. Thus $\text{Ind}_{G}(\alpha)=\text{Ind}(\alpha)$.

\end{proof}

\bigskip

\noindent \textbf{Proof of Corollary \ref{classificationcor}}. Suppose $\alpha:A^{G}\rightarrow A^{G}$ is a G-duality such that $\text{Ind}(\alpha)=1$ and $\text{DHR}(\alpha)=\text{Id}_{\mathcal{Z}(\text{Rep}(G))}$. Then by Theorem \ref{mainthm} there is a lift of $\alpha$ to an sQCA $\widetilde{\alpha}$ which can be chosen by \ref{thm:sQCA} to have trivial $H^{2}$ index. Furthermore, by Proposition \ref{indexlemma}, $\text{Ind}(\widetilde{\alpha})=1$, so $\widetilde{\alpha}$ is a finite-depth circuit. 

Now, for abelian $G$ in the regular representation, it was recently shown that globally symmetric circuits with trivial $H^{2}$ invariant are locally symmetric \cite[Theorem 1]{ma2024quantumcellularautomatasymmetric}. 
More generally, the corresponding statement for groups $G$, which are not necessarily abelian, is shown (implicitly\footnote{The argument stated there uses ancilla but can easily be modified to avoid this step by coarse-graining and applying Fell's absorption principle in the regular representation}) in \cite{long2024topologicalphasesmanybodylocalized}. This concludes the proof.

\section*{Acknowledgments} 
CJ and KS were supported by NSF DMS-2247202. 
DJW was supported in part by the Australian Research Council Discovery Early Career Research Award (DE220100625). 
This paper grew out of discussions that took place during the ``QCA group" at the workshop ``Fusion categories and tensor networks" at the American Institute of Mathematics in 2022. 
After completing the first arXiv version of this work we became aware of the related work \cite{ma2024quantumcellularautomatasymmetric}. 
We have now included comments on the relationship between our results and those of~\cite{ma2024quantumcellularautomatasymmetric}. We also thank Dominic Else for pointing out how the results of~\cite{long2024topologicalphasesmanybodylocalized} can be used to obtain Corollary \ref{classificationcor}, which is a stronger statement than in earlier versions.

\bigskip

\bigskip

\bibliographystyle{alpha}
\bibliography{Reference}

\end{document}